\documentclass[11pt]{amsart}
\usepackage{fullpage}
\usepackage[foot]{amsaddr}
\usepackage{microtype}
\usepackage[OT1]{fontenc}
\usepackage{amsmath}
\usepackage{amssymb}
\usepackage{amsthm}
\usepackage{thmtools}
\makeatletter
\@ifundefined{newcounteralias}{}{\renewcommand\thmt@autorefsetup{\@xa\def\csname\thmt@envname autorefname\@xa\endcsname\@xa{\thmt@thmname}}}
\makeatother
\usepackage[linesnumbered,boxed,ruled,vlined]{algorithm2e}
\usepackage{algpseudocode}
\usepackage{enumitem}
\usepackage{xifthen}
\usepackage{xspace}

\usepackage[margin=1cm]{caption}
\usepackage{subfig}

\usepackage[thinlines]{easytable}

\usepackage[bookmarks=true,hypertexnames=false,pagebackref]{hyperref}
\hypersetup{colorlinks=true, citecolor=blue, linkcolor=red, urlcolor=blue}

\usepackage{pgfplots}
\pgfplotsset{compat=1.16}
\usepackage{tikz}
\usetikzlibrary{arrows,arrows.meta,backgrounds,calc,fit,decorations.pathreplacing,decorations.markings,shapes.geometric}

\tikzstyle{internal} = [draw, fill, shape=circle]
\tikzstyle{external} = [shape=circle]
\tikzstyle{square}   = [draw, fill, rectangle]
\tikzstyle{triangle} = [draw, fill, regular polygon, regular polygon sides=3, inner sep=3pt]
\tikzstyle{pentagon} = [draw, fill, regular polygon, regular polygon sides=5, inner sep=2pt, minimum size=14pt]
\tikzset{every fit/.append style=text badly centered}

\usetikzlibrary{positioning,chains,fit,shapes,calc}
\usetikzlibrary{trees}
\usetikzlibrary{decorations.pathreplacing}
\usetikzlibrary{decorations.pathmorphing}
\usetikzlibrary{decorations.markings}
\tikzset{>=latex}

\usepackage{ifthen}

\usepackage{cleveref}

\usepackage[textsize=tiny]{todonotes}

\usepackage[normalem]{ulem}

\usepackage{mleftright}

\usepackage{cool}
\Style{DSymb={\mathrm d},DShorten=true,IntegrateDifferentialDSymb=\mathrm{d}}

\newcommand{\Ex}{\mathop{\mathbb{{}E}}\nolimits}
\renewcommand{\Pr}{\mathop{\mathrm{Pr}}\nolimits}

\def\*#1{\mathbf{#1}}
\def\+#1{\mathcal{#1}}
\def\-#1{\mathrm{#1}}
\def\=#1{\mathbb{#1}}
\def\^#1{\mathbb{#1}}

\newcommand{\norm}[2]{\ensuremath{\Vert #2 \Vert_{#1}}}
\newcommand{\abs}[1]{\ensuremath{\left\vert#1\right\vert}}

\newcommand{\eps}{\varepsilon}

\newcommand{\Diri}[2]{\ensuremath{\+E_{#1}\left(#2\right)}}

\newcommand{\Var}[2]{\ensuremath{\textnormal{Var}_{#1}\left(#2\right)}}

\newcommand{\diag}{\operatorname{diag}}

\newtheorem{theorem}{Theorem}

\newtheorem{lemma}[theorem]{Lemma}

\newtheorem{proposition}[theorem]{Proposition}
\newtheorem{corollary}[theorem]{Corollary}
\theoremstyle{definition}

\newtheorem{definition}[theorem]{Definition}
\newtheorem{example}{Example}
\theoremstyle{remark}
\newtheorem*{remark}{Remark}

\crefname{theorem}{Theorem}{Theorems}
\crefname{observation}{Observation}{Observations}
\crefname{claim}{Claim}{Claims}
\crefname{condition}{Condition}{Conditions}
\crefname{algorithm}{Algorithm}{Algorithms}
\crefname{property}{Property}{Properties}
\crefname{example}{Example}{Examples}
\crefname{fact}{Fact}{Facts}
\crefname{lemma}{Lemma}{Lemmas}
\crefname{corollary}{Corollary}{Corollaries}
\crefname{definition}{Definition}{Definitions}
\crefname{remark}{Remark}{Remarks}
\crefname{proposition}{Proposition}{Propositions}
\crefname{equation}{equation}{equations}
\crefname{enumi}{Case}{Case}
\creflabelformat{enumi}{(#2#1#3)}

\makeatletter
\def\prob#1#2#3{\goodbreak\begin{list}{}{\labelwidth\z@ \itemindent-\leftmargin
      \itemsep\z@  \topsep6\p@\@plus6\p@
      \let\makelabel\descriptionlabel}
  \item[\textbf{Name}]#1
  \item[\textbf{Instance}]#2
  \item[\textbf{Output}]#3
  \end{list}}
\makeatother

\makeatletter
\providecommand\@dotsep{5}
\def\listtodoname{Todo list}
\def\listoftodos{\@starttoc{tdo}\listtodoname}
\makeatother

\usepackage{nicefrac,comment}

\newcommand{\Cov}{\operatorname{Cov}}

\newcommand{\gap}{\operatorname{gap}}
\newcommand{\one}{\mathbf{1}}

\numberwithin{equation}{section}
\title{Degree-Free Spectral Independence \\ for Log-Concave Holant Measures}
\author{Xiaoyu Chen}
\address[Xiaoyu Chen]{\textnormal{\textit{Email}: \texttt{xiaoyu@mit.edu}.} Massachusetts Institute of Technology}
\thanks{Xiaoyu Chen is supported by the NSF CAREER grant CCF2443045 and the Reed Fund at MIT}
\author{Zejia Chen}
\address[Zejia Chen]{\textnormal{\textit{Email}: \texttt{zchen3091@gatech.edu}.} Georgia Institute of Technology}
\author{Xinyuan Zhang}
\address[Xinyuan Zhang]{\textnormal{\textit{Email}: \texttt{zhangxy@smail.nju.edu.cn}.} Nanjing University}
\date{}
\begin{document}

\begin{abstract}
We establish a degree-independent bound on spectral independence for log-concave Holant problems on simple graphs. As a corollary, we obtain relaxation-time bounds for Glauber dynamics of $O_\lambda(m)$ for the monomer--dimer model at activity $\lambda$ and $O_{b,\lambda}(m)$ for $b$-matchings at fugacity $\lambda>0$, where $m$ is the number of edges. For uniform $b$-matchings, the relaxation-time bound improves to $O(bm)$.
The main proof ideas were found using GPT-5.6 Sol.
\end{abstract}
\maketitle

\tableofcontents

\newpage
\section{Introduction}\label{sec:model}

Sampling combinatorial structures is a central problem in randomized algorithms. Many models on graphs can be naturally described in terms of random variables associated with either vertices or edges. In the first class, configurations assign states to vertices, with interactions between adjacent vertices; familiar examples include the hard-core model~\cite{LV97}, the Ising model~\cite{Fos59,EFH60}, and proper graph colorings~\cite{Jer95}. In the second class, configurations assign states to edges, with local interactions or constraints involving the edges incident to a common vertex. Examples include the monomer--dimer model~\cite{JS89}, $b$-matchings~\cite{HLZ16}, $b$-edge covers~\cite{HLZ16}, and the even-subgraph model~\cite{JS93}. The Holant framework provides a common language for this second class, expressing a configuration's weight as a product of local functions at the vertices. In this paper, we study efficient sampling algorithms within this framework, focusing on Holant measures with log-concave local signatures.

Let $G=(V,E)$ be a finite simple undirected graph, with $d_v=\deg_G(v)$ and $\Delta=\max_v d_v$. A \emph{binary symmetric Holant instance} $(G,\mathbf f,\boldsymbol\lambda)$ consists of positive edge activities $\boldsymbol\lambda=(\lambda_e)_{e\in E}$ and nonnegative signatures $\mathbf f=(f_v)_{v\in V}$, where $f_v:\{0,\ldots,d_v\}\to\mathbb R_{\ge0}$. Its partition function $Z=Z_{G,\mathbf f,\boldsymbol\lambda}$ and probability measure $\mu=\mu_{G,\mathbf f,\boldsymbol\lambda}$ are
\begin{equation}\label{eq:holant}
 Z=\sum_{S\subseteq E}\prod_vf_v(D_v(S))\prod_{e\in S}\lambda_e \quad \text{and} \quad
 \mu(S)=\frac1Z\prod_vf_v(D_v(S))\prod_{e\in S}\lambda_e,
\end{equation}
where $D_v(S)$ denotes the number of edges in $S$ incident to $v$.

Following Chen and Gu~\cite{CG24}, we focus on log-concave signatures, defined as follows.

\begin{definition}[Log-concave signatures]\label{def:lc}
Let $f=[f(0),f(1),\ldots,f(d)]$ be a sequence with nonnegative entries. We call $f$ a \emph{log-concave signature} when both of the following hold:
\begin{enumerate}[label=(\alph*)]
\item \emph{Log-concavity:} for each $k\in\{1,\ldots,d-1\}$, we have $f(k)^2\ge f(k-1)f(k+1)$.
\item \emph{No internal zeros:} for any $0\le i<j\le d$ with $f(i),f(j)>0$, every entry $f(k)$ with $i\le k\le j$ is positive.
\end{enumerate}
\end{definition}

Symmetric Holant problems with log-concave signatures encompass a wide range of models.
We illustrate this framework with matching models, which are the main applications in this paper.

\begin{example}[Matching models]\label{example:matching-cover}
Let $G=(V,E)$ be a simple graph.
Fix a common edge activity $\lambda>0$, so that
$\boldsymbol{\lambda}=\lambda\mathbf{1}$.
\begin{itemize}
   \item \emph{The monomer--dimer model.}
    Taking $b_v=1$ for every $v\in V$ gives the monomer--dimer
    model at fugacity $\lambda$, namely the distribution on ordinary
    matchings $M$ with probability proportional to $\lambda^{|M|}$.

    \item \emph{$\mathbf{b}$-matchings.}
    Given integer capacities $\mathbf{b}=(b_v)_{v\in V}$ for each vertex $v\in V$, define the signatures
    $$
        f_v(i)=\mathbf{1}[i\le b_v], \qquad 0\le i\le d_v.
    $$
    This induces the Gibbs measure
    $$
      \mu(S)\propto \mathbf{1}[S\in\mathcal{M}_{\mathbf{b}}]\lambda^{|S|},
    $$
    where
    $$
      \mathcal{M}_{\mathbf{b}} :=\{S\subseteq E : D_v(S)\le b_v, \forall v\in V\}.
    $$
    In particular, $\lambda=1$ gives the uniform measure on
    $\mathcal{M}_{\mathbf{b}}$.

\end{itemize}
\end{example}

A standard approach to sampling is the Markov chain Monte Carlo (MCMC) method, which simulates an ergodic Markov chain whose stationary distribution is the target measure. A basic example is Glauber dynamics, which at each step selects a variable uniformly at random and resamples it from its conditional distribution given all the other variables. The efficiency of this approach is governed by the mixing time: the number of steps required for the chain to be within a prescribed total variation distance of stationarity, starting from any initial state~\cite{LP17}. In recent years, spectral independence, introduced by Anari, Liu, and Oveis Gharan~\cite{ALO20}, has emerged as a general framework for establishing efficient sampling algorithms~\cite{CLV20,CLV21,CFYZ21,CFYZ22,CE22,CZ23,CG24,CCYZ25,CCCYZ25,CSV26}. The central idea is to establish weak correlations among variables, quantified by spectral bounds on influence matrices, and then use local-to-global arguments to derive rapid mixing of the Glauber dynamics. We recall the definition of spectral independence below.

\begin{definition}[Spectral independence~\cite{ALO20}]\label{def:si}
For any feasible pinning $\tau$, define the influence matrix $\Phi_\mu^\tau$ on the remaining non-deterministic edges by
\begin{equation}\label{eq:influence}
 \Phi_\mu^\tau(e,f)=\mu^\tau(f\in S\mid e\in S)-\mu^\tau(f\in S\mid e\notin S)
 \quad(e\ne f),\qquad \Phi_\mu^\tau(e,e)=0.
\end{equation}
The law $\mu$ is $\eta$-\emph{spectrally independent} if $\lambda_{\max}(\Phi_\mu^\tau)\le\eta$ for every feasible pinning $\tau$.
\end{definition}

For log-concave Holant measures on bounded-degree graphs, Chen and Gu~\cite{CG24} established constant spectral independence by proving coupling independence via a recursive coupling, building on the approach of Chen and Zhang~\cite{CZ23}. Their spectral independence bound is as follows.
\begin{theorem}[Chen and Gu~{\cite[Theorem~13]{CG24}}]\label{thm:cg}
Let $(G,\mathbf f,\boldsymbol\lambda)$ be a Holant instance on a graph of maximum degree $\Delta$ with $E\ne\varnothing$, positive edge activities, and log-concave signatures satisfying $f_v(0)>0$. Define the normalized generating polynomials and parameters
\begin{equation}\label{eq:cg-parameters}
 \begin{gathered}
 P_{f_v}(x)=\frac1{f_v(0)}\sum_{k=0}^{d_v}\binom{d_v}{k}f_v(k)x^k,\\
 r_{\max}=\max_{v:d_v\ge1}\frac{f_v(1)}{f_v(0)},\qquad
 \lambda_{\max}=\max_{e\in E}\lambda_e,\qquad
 P_{\max}=\max_{v\in V}P_{f_v}(r_{\max}\lambda_{\max}).
 \end{gathered}
\end{equation}
Then $\mu$ is $2(P_{\max}-1)$-spectrally independent. In particular, after every feasible pinning $\tau$,
\begin{equation}\label{eq:cg-si}
 \lambda_{\max}(\Phi_\mu^\tau)\le2(P_{\max}-1)
 \le2\bigl((1+r_{\max}^2\lambda_{\max})^\Delta-1\bigr).
\end{equation}
\end{theorem}

Their result yields spectral independence bounds for a range of matching models. For the monomer--dimer model at activity $\lambda>0$, \cref{thm:cg} gives a spectral independence bound of $2\Delta\lambda$. Together with their marginal bounds and the local-to-global argument, this yields an $O_{\Delta,\lambda}(n\log n)$ mixing time for Glauber dynamics, where $n=|V|$~\cite[Theorem~1]{CG24}. For uniform $b$-matchings, their spectral independence bound is $2\sum_{k=1}^{\min\{b,\Delta\}}\binom{\Delta}{k}=O_b(\Delta^b)$, and the corresponding mixing-time bound is $O_\Delta(n\log n)$~\cite[Theorem~1 and Remark~22]{CG24}.

Despite their success on bounded-degree graphs, the coupling independence, or total influence (i.e., $\norm{\infty}{\Phi_\mu}$), inherently depends on the maximum degree, even for the monomer--dimer model. On the infinite $\Delta$-regular tree, the coupling independence and total influence are $\Theta_\lambda(\sqrt{\Delta})$ at fixed activity $\lambda>0$~\cite[Remark~6.5]{CLV21}, and finite truncations inherit this lower bound as their depth grows.

In this paper, we establish the following degree-independent spectral independence bounds for log-concave Holant measures.
Following the vertex-threshold notation in \cite{CG24}, let $b_v=\max\{k:f_v(k)>0\}$ and $\mathbf b=(b_v)_{v\in V}$. Our general theorem assumes $b_v\le b$ and bounds on the insertion odds
\begin{equation}\label{eq:ratios}
 r_v(k)=\frac{f_v(k+1)}{f_v(k)}\quad(0\le k<b_v),\qquad r_v(b_v)=0,
 \qquad \theta_{uv}(i,j)=\lambda_{uv}r_u(i)r_v(j).
\end{equation}
These odds are intrinsic: they remain unchanged by $f_v(k)\mapsto a_v^kf_v(k)$ and $\lambda_{uv}\mapsto\lambda_{uv}/(a_ua_v)$.

\begin{theorem}[Degree-free Holant spectral independence]\label{thm:holant}
Let $b\ge1$ be an integer and $0<\gamma\le\Theta<\infty$. Suppose $G$ is simple, every signature is log-concave with $f_v(0)>0$ and $b_v\le b$, and
\begin{equation}\label{eq:odds-assumption}
 \gamma\le\theta_{uv}(i,j)\le\Theta
 \qquad(uv\in E,\ 0\le i<b_u,\ 0\le j<b_v).
\end{equation}
Set
\begin{equation}\label{eq:main-constant}
 L=L_b(\Theta)=\max\{\Theta,\Theta^b\},\qquad
 T=T_b(\gamma,\Theta)=\frac{L_b(\Theta)}{\min\{1,\gamma\}}.
\end{equation}
There is an absolute constant $C$ such that, after every feasible pinning $\tau$,
\begin{equation}\label{eq:holant-main}
 \lambda_{\max}(\Phi_\mu^\tau)\le C(1+\Theta)\bigl[1+\Theta+(b-1)T^2\bigr].
\end{equation}
The bound is independent of the underlying graph and the maximum degree.
\end{theorem}

\begin{remark}
For the monomer--dimer model at activity $\lambda>0$, \cref{thm:holant} gives $O((1+\lambda)^2)$ spectral independence on every finite simple graph: take $b=1$ and $\gamma=\Theta=\lambda$. This bound holds after every feasible pinning and is independent of the maximum degree for each fixed activity. It answers the question posed by Chen, Liu, and Vigoda~\cite[Section~8]{CLV21} of improving the spectral independence bound for the monomer--dimer model beyond the bound obtained by controlling total influence.
\end{remark}

As a corollary, we have the following mixing results for matching models.

\begin{corollary}[Matching models]\label{cor:matchings}
There is an absolute constant $C$ such that the following holds on every finite simple graph with $n=|V|$ and $m=|E|\ge1$, for $0<\varepsilon<1/2$.
\begin{enumerate}[label=(\roman*)]
\item \emph{Monomer--dimer model.} At fugacity $\lambda>0$, the measure is $C(1+\lambda)^2$-spectrally independent. Its Glauber dynamics satisfies
\begin{equation}\label{eq:monomer-relaxation}
 t_{\mathrm{rel}}=O_\lambda(m), \qquad
 t_{\mathrm{mix}}(\varepsilon)=O_\lambda\!\left(m\left(n\log n+\log\frac1\varepsilon\right)\right).
\end{equation}
\item \emph{$b$-matchings.} At fugacity $\lambda>0$, the measure with integer capacities $0\le b_v\le b$ and $b\ge1$ is $C(1+\lambda)\bigl[1+\lambda+(b-1)\max\{1,\lambda^{2b}\}\bigr]$-spectrally independent. Its Glauber dynamics satisfies
\begin{equation}\label{eq:b-matching-relaxation}
 t_{\mathrm{rel}}=O_{b,\lambda}(m),\qquad
 t_{\mathrm{mix}}(\varepsilon)=O_{b,\lambda}\!\left(m\left(\min\{m,bn\log n\}+\log\frac1\varepsilon\right)\right).
\end{equation}
In particular, uniform $b$-matchings are $Cb$-spectrally independent, and their Glauber dynamics satisfies
\begin{equation}\label{eq:uniform-relaxation}
 t_{\mathrm{rel}}=O(bm),\qquad
 t_{\mathrm{mix}}(\varepsilon)=O\!\left(bm\left(\min\{m,bn\log n\}+\log\frac1\varepsilon\right)\right).
\end{equation}
\end{enumerate}
\end{corollary}
\begin{remark}
The bounds for uniform $b$-matchings also apply to uniform
$b$-edge covers by taking the complement map
$S\mapsto E\setminus S$.
Formally, given integer lower bounds $0\le b_v\le d_v$ and
fugacity $\lambda>0$, the measure on $b$-edge covers
$M\subseteq E$ satisfying $D_v(M)\ge b_v$ for $v\in V$,
with probability proportional to $\lambda^{|M|}$, is
$C_{B,\lambda}$-spectrally independent, where
$B=\max\{1,\max_{v\in V}(d_v-b_v)\}$.
\end{remark}

The proof of \cref{cor:matchings} is deferred to Appendix~\ref{app:matching-mixing}.

\subsection{Proof overview}\label{sec:overview}

Our goal is to bound the covariance of the edge indicators by a diagonal matrix. We first seek a comparison that separates a diagonal contribution from one carried by the vertices. The edge-to-vertex estimate below does exactly this: its only additional term consists of weighted sums of edge coefficients at each vertex. These sums single out the occupied degrees $D_v$ as the observables through which to control the remaining contribution. Since degrees are sums of edge indicators, a Schur complement of their joint covariance connects the two variance bounds.

The needed degree variance bound comes from the recursive coupling of Chen and Gu~\cite[Section~3.1]{CG24}: although an alternating disagreement trail may contain many edges, its contributions to occupied degrees cancel at internal vertices. This cancellation gives degree coupling independence~\cite{CZ23}, and hence a variance bound for additive degree observables that is independent of the maximum degree. This use of vertex observables is also motivated by earlier works on monomer dynamics~\cite{AASV21,AJKPV22}.

Fix a feasible pinning and work in the residual instance, omitting deterministic edges and vertices with no live incident edges.
Write $X_e=\one_{\{e\in S\}}$, $p_e=\Ex X_e$, $D_v=\sum_{e\ni v}X_e$, and $q_v=\Pr(D_v<b_v)>0$, where $b_v$ is the residual support bound.
We now introduce the following estimates that lead to the bound for spectral independence.

\begin{lemma}[Vertex variance bound]\label{lem:overview-degree}
For arbitrary real functions $g_v$ of the occupied degree,
\[
 \Var{}{\sum_vg_v(D_v)}\le2\sum_v\Var{}{g_v(D_v)},
\]
where $D_v$ denotes the number of selected edges incident to vertex $v$.
In particular, for every real vector $h$,
\begin{equation}\tag{A}\label{eq:overview-degree}
 \Var{}{\sum_vh_vD_v}\le2b^2\sum_vq_vh_v^2.
\end{equation}
\end{lemma}

This lemma bounds the variance of an additive degree observable by twice the sum of its individual variances, independently of the maximum degree. It follows from degree coupling independence and the variance characterization of spectral independence; see \cref{sec:degrees}.

\begin{lemma}[Edge-to-vertex variance bound]\label{lem:overview-transfer}
For every real edge vector $a$,
\begin{equation}\tag{B}\label{eq:overview-source}
 \Var{}{\sum_ea_eX_e}\le C_1\sum_ep_ea_e^2+C_2\sum_v\frac{\left(\sum_{e\ni v}p_ea_e\right)^2}{q_v},
\end{equation}
where $C_1=O(1+\Theta+(b-1)T^2)$ and $C_2=O(1)$, with absolute implied constants.
\end{lemma}

This lemma reduces edge variance to a diagonal contribution and weighted sums of edge coefficients at vertices. Its proof bounds the local star variances of endpoint-degree surrogates and combines these estimates by revealing degrees at random; see its proof in \cref{sec:stars}.

Before combining the two variance bounds, we recall the standard Schur-complement criterion. For real symmetric matrices, write $A\succeq0$ when $A$ is positive semidefinite, and $A\preceq B$ when $B-A\succeq0$.
This criterion is well known; see~\cite[Sections~2.3 and~4.1]{CL26} for a proof.

\begin{lemma}[Schur-complement criterion]\label{lem:schur-complement}
Let $A,C$ be real symmetric matrices and $B$ a real matrix of compatible dimensions. If $C$ is positive definite, then
\[
 \begin{pmatrix}A&B\\B^\top&C\end{pmatrix}\succeq0
 \quad\Longleftrightarrow\quad
 A-BC^{-1}B^\top\succeq0.
\]
The matrix $A-BC^{-1}B^\top$ is called the \emph{Schur complement} of $C$ in this block matrix.
\end{lemma}

We apply this criterion to the joint covariance of edge indicators and occupied degrees.

\begin{proof}[Proof of \cref{thm:holant} via the Schur complement]
Fix a feasible pinning $\tau$ and work in the residual instance. Pinning shifts and truncates the signatures, preserving log-concavity, positivity at zero, and the bounds $\gamma,\Theta,b$. Omit deterministic edges and vertices with no remaining incident edges; if no edges remain, the conclusion is immediate. Every retained $p_e$ is positive, and $q_v>0$ because the empty residual configuration has positive weight. Write $X=(X_e)_{e\in E}$ and $D=(D_v)_{v\in V}$ as column vectors, indexed by the remaining edges and vertices.

To control the edge covariance using the degree variance bound, we begin with the joint covariance of $X$ and $D$. Write $\Sigma=\Cov(X)$, $Q=\diag(q_v)$, and let $H=(\one_{\{e\ni v\}})_{v\in V,e\in E}$ be the unsigned vertex-edge incidence matrix. Since $D=HX$, the joint covariance is
\[
 \Cov\!\begin{pmatrix}X\\D\end{pmatrix}
 =\begin{pmatrix}
 \Sigma & \Sigma H^\top\\
 H\Sigma & H\Sigma H^\top
 \end{pmatrix}
 \succeq0.
\]
By \eqref{eq:overview-degree}, the degree block satisfies
\begin{equation}\tag{$\mathrm{A}'$}\label{eq:schur-degree}
 H\Sigma H^\top\preceq2b^2Q.
\end{equation}
We may therefore increase that block to $2b^2Q$ while preserving positive semidefiniteness:
\[
 \begin{pmatrix}
 \Sigma & \Sigma H^\top\\
 H\Sigma & 2b^2Q
 \end{pmatrix}
 \succeq0.
\]
Since $2b^2Q$ is positive definite, \cref{lem:schur-complement} with $A=\Sigma$, $B=\Sigma H^\top$, and $C=2b^2Q$ gives
\begin{equation}\label{eq:schur-mixed}
 \Sigma H^\top Q^{-1}H\Sigma\preceq2b^2\Sigma.
\end{equation}
This bounds the mixed covariances between edges and degrees in terms of the edge covariance itself.

We now use this estimate to control the vertex contribution in \eqref{eq:overview-source}. Write $P=\diag(p_e)$ and take nonnegative constants $C_1,C_2$ from that estimate. Since $(HPa)_v=\sum_{e\ni v}p_ea_e$ for every $a\in\mathbb R^E$, its matrix form is
\begin{equation}\tag{$\mathrm{B}'$}\label{eq:schur-edge}
 \Sigma\preceq C_1P+C_2PH^\top Q^{-1}HP.
\end{equation}
Our goal is to bound $\Sigma$ by a multiple of $P$. Since $P$ is positive definite, we may multiply \eqref{eq:schur-edge} on the left by $\Sigma P^{-1}$ and on the right by $P^{-1}\Sigma$. These multipliers are transposes, so the positive semidefinite order is preserved. Using \eqref{eq:schur-mixed}, we obtain
\begin{align*}
 \Sigma P^{-1}\Sigma P^{-1}\Sigma
 &\preceq C_1\Sigma P^{-1}\Sigma
     +C_2\Sigma H^\top Q^{-1}H\Sigma
 &\preceq C_1\Sigma P^{-1}\Sigma+2C_2b^2\Sigma.
\end{align*}
Let $t$ be any generalized eigenvalue of $(\Sigma,P)$, meaning that $\Sigma a=tPa$ for some nonzero $a\in\mathbb R^E$. Since $\Sigma$ is positive semidefinite and $P$ is positive definite, $t\ge0$. Evaluating the preceding inequality on $a$ and dividing by $a^\top Pa>0$ gives
\[
 t^3\le C_1t^2+2C_2b^2t.
\]
Consequently, every generalized eigenvalue satisfies
\[
 t\le\rho:=\frac{C_1+\sqrt{C_1^2+8C_2b^2}}{2}
 \le C_1+\sqrt{2C_2}\,b
 =O\bigl(1+\Theta+(b-1)T^2\bigr),
\]
where the last step uses $T\ge1$. By the generalized spectral theorem, $\Sigma\preceq\rho P$.

Finally, the insertion-odds bound gives $p_e\le\Theta(1-p_e)$, so
\[
 \Sigma\preceq\rho P
 \preceq(1+\Theta)\rho\,\diag\bigl(p_e(1-p_e)\bigr).
\]
By the variance characterization of spectral independence, this proves the required bound for every feasible pinning.
\end{proof}

An alternative proof of~\Cref{thm:holant} using variance maximization is given in Appendix~\ref{app:alternative-proof}.
\subsection{Related work}\label{sec:related-work}

For the monomer--dimer model, the foundational work of Jerrum and Sinclair~\cite{JS89,Sin92} used the canonical path method and established a mixing-time bound of $O_\lambda(n^2m\log n)$, where $n=|V|$, $m=|E|$, and $\lambda>0$ is the activity; see also~\cite{Jer03}. Subsequently, \cite{BGKNT07} established correlation decay through a two-level analysis of the computation-tree recursion, obtaining a contraction factor $1-\Theta_\lambda(1/\sqrt{\Delta})$ for fixed $\lambda$. Building on this contraction, \cite{CLV21} proved an $O_{\Delta,\lambda}(1)$ spectral independence bound and, through entropy factorization, obtained the optimal $O_{\Delta,\lambda}(m\log n)$ mixing time for single-edge Glauber dynamics on bounded-degree graphs. However, the constant in this mixing-time bound grows exponentially with the maximum degree. To obtain polynomial dependence on the degree, \cite{CFJMYZ25} developed a transport flow framework combining canonical paths with local-to-global arguments for functional inequalities. This yields a $\widetilde{O}_\lambda(\Delta^2m)$ mixing time for the lazy Jerrum--Sinclair chain, which permits exchanges of adjacent edges as well as single-edge insertions and deletions, and, by comparison, a $\widetilde{O}_\lambda(\Delta^3m)$ mixing time for single-edge Glauber dynamics. Concurrently, \cite{CCCYZ25} used a trickle-down theorem for field dynamics to obtain further bounds on arbitrary simple graphs: $O_\lambda(mn\log n)$ mixing for every fixed $0<\lambda<1$, and $O(\sqrt{\Delta}\,mn\log n)$ mixing for uniform matchings ($\lambda=1$).

Moving beyond ordinary matchings, sampling $b$-matchings brings more general local degree constraints into consideration. McQuillan~\cite{McQ13} introduced windable functions as a way to construct multicommodity flows for Holant problems. Building on this approach, Huang, Lu, and Zhang~\cite{HLZ16} proved polynomial mixing of a chain on an enlarged space of consistent and nearly consistent half-edge assignments, yielding an FPRAS for $b$-matchings with $b\le7$ on general graphs. Chen and Gu~\cite{CG24} subsequently obtained $O_{\Delta,\lambda}(n\log n)$ mixing of single-edge Glauber dynamics for arbitrary vertex capacities on all bounded-degree graphs.

The main purpose of this paper is to establish spectral independence without requiring bounded total influence. Two approaches to this problem have been explored in the literature. The first is the trickle-down method, which propagates spectral estimates from more conditioned distributions to less conditioned ones. Originating in the study of local spectral expansion~\cite{Opp18,ALOV19}, this method has been successfully applied in a variety of sampling settings~\cite{ALG21,WZZ24,AKV24,CCCYZ25,KOT26}. The second is the matrix inverse approach, which uses inverses of local matrices as building blocks for an exact or approximate inverse of the full correlation matrix, thereby reducing the spectral analysis to local computations~\cite{CYYZ24}.

\subsection{Concurrent work}\label{sec:concurrent-work}

We have learned that, in concurrent and independent work, Chihao Zhang and Zihan Zhang also established a constant spectral independence bound for the monomer--dimer model on graphs of unbounded degree, where the bound depends only on the activity.

\subsection{AI disclosure}
We used AI tools, including GPT-5.6 Sol, to develop the main proof ideas and to edit the manuscript. The authors are responsible for the mathematical content and the final presentation.

\section{Preliminaries}
\subsection{Basic Markov Chains}
Let $\Omega$ be a finite state space.
A time-homogeneous Markov chain $(X_t)_{t\in \mathbb{Z}_{\geq 0}}$ is specified by a \emph{transition matrix} $P\in \mathbb{R}_{\geq 0}^{\Omega \times \Omega}$, where $P(x,y)$ is the probability of moving from $x$ to $y$ in one step, and an initial state $X_0$.
Writing $\mu_t$ for the distribution of $X_t$, we have that $\mu_{t+1}=\mu_tP$.
Note that by definition, $P$ is a stochastic matrix: $\sum_{y\in\Omega}P(x,y)=1$ for any $x\in\Omega$.

The Markov chain is \emph{irreducible} if for any $x,y\in \Omega$, there exists $t > 0$ such that $P^t(x,y) > 0$ and the Markov chain is \emph{aperiodic} if for any $x\in \Omega$, $\gcd\{t: P^t(x,x) > 0\} = 1$.
The chain is called \emph{ergodic} if it is both irreducible
and aperiodic.
By the fundamental theorem of finite Markov chains, an ergodic Markov chain has a unique \emph{stationary distribution} $\mu$, which satisfies $\mu P = \mu$.
We say $P$ is reversible with respect to $\mu$ if  it satisfies the detailed balance equations: for any $x,y\in \Omega$, $\mu(x)P(x,y) = \mu(y)P(y,x)$.
Detailed balance implies stationarity.

A canonical example of a Markov chain is \emph{Glauber dynamics}.
For a Holant measure $\mu=\mu_{G,\boldsymbol{f},\boldsymbol{\lambda}}$ on a simple graph $G = (V,E)$, Glauber dynamics updates the occupancy of a single edge at each step.
Formally, for $t\geq 0$, we obtain $X_{t+1}$ from $X_t$ as follows:
\begin{itemize}
  \item Pick $e\in E$ uniformly at random and let $X' = X_{t} \setminus \{e\}$.
  \item Set $X_{t+1} = X' \cup \{e\}$ with probability $\frac{\mu(X'\cup\{e\})}{\mu(X')+\mu(X'\cup\{e\})}$ and set $X_{t+1} = X'$ with the remaining probability.
\end{itemize}
It is easy to verify that this chain is reversible with respect to $\mu$, and hence $\mu$ is stationary.

The efficiency of a Markov chain sampling algorithm is governed by
its rate of convergence to stationarity. For a finite ergodic chain with transition matrix $P$ and stationary
distribution $\mu$, its \emph{mixing time} is defined, for
$0<\varepsilon<1/2$, by
$$
  t_{\-{mix}}(\eps) := \max_{x\in \Omega} \min \{t\geq 0\mid d_{\-{TV}}(P^t(x,\cdot),\mu) \leq \eps\},
$$
where $d_{\-{TV}}(\mu,\nu) = \frac 12 \sum_{x\in \Omega}\abs{\mu(x) - \nu(x)}$ denotes the total variation distance.

For reversible chains, one way to bound the mixing time is to obtain a lower bound on the \emph{absolute spectral gap} of the transition matrix.
Suppose $P$ is reversible and its eigenvalues are $1=\lambda_1 \geq \lambda_2 \geq \cdots \geq \lambda_{|\Omega|}$.
The spectral gap $\gap(P)$ and \emph{absolute spectral gap} $\gamma_*(P)$ are defined respectively by
$$
  \gap(P) := 1 - \lambda_2,\qquad \gamma_*(P) := 1 - \max_{2\leq i\leq |\Omega|}\abs{\lambda_i}.
$$
Note that irreducibility implies $\lambda_2<1$, and irreducibility together with aperiodicity implies $\gamma_*(P)>0$.
The following proposition gives an upper bound on the mixing time in terms of $\gamma_*(P)$.
\begin{proposition}\label{prop:spectral-gap_mixing-time}
  Suppose $(X_t)_{t\geq 0}$ is an ergodic Markov chain on a finite state space $\Omega$ with transition matrix $P$.
  $P$ is reversible with respect to its stationary distribution $\mu$.
  Then, for any $\varepsilon\in(0,\frac 12)$, we have that
  $$
    t_{\-{mix}}(\eps) \leq \frac {1}{\gamma_*(P)}\log \frac{1}{\eps\cdot \mu_{\min}},
  $$
  where $\mu_{\min} = \min \{\mu(x)\mid \mu(x) > 0, x\in \Omega\}$.
\end{proposition}
The \emph{relaxation time} is given by the multiplicative inverse of the absolute spectral gap: $t_{\-{rel}} = \frac 1{\gamma_*(P)}$.
For Glauber dynamics, since $P$ is positive semi-definite, we have that $\gap(P_{\-{GD}}) = \gamma_*(P_{\-{GD}})$.

\subsection{Field dynamics}\label{sec:field-dynamics}
For an edge-set law $\mu$ and a feasible occupied set $S$, write $\mu^S$ for the law of the remaining edges conditioned on $S$ being occupied. Define its uniform tilt by
\[
 (\lambda*\mu^S)(J)\propto \lambda^{|J|}\mu^S(J),\qquad \lambda>0.
\]
The following field dynamics comparison turns covariance bounds for these tilted conditional laws into a spectral-gap bound for Glauber dynamics.

\begin{theorem}[Field dynamics comparison~{\cite[Theorem~1.16 and Lemma~1.18]{CCCYZ25}}]\label{thm:field-comparison}
Let $\mu$ be a probability law on a downward closed family $\mathcal{X} \subseteq 2^E$. Let $P_{\mathrm{GD}}$ be the Glauber dynamics on $\mu$, and define
\[
 r_{\mathrm{ins}}=
 \max_{\substack{J\subseteq E,\ e\notin J\\\mu(J\cup\{e\})>0}}
 \frac{\mu(J\cup\{e\})}{\mu(J)}.
\]
Suppose a nonnegative measurable function $c:(0,1)\to\mathbb R$ satisfies, for every feasible occupied set $S$ and every $0<\eta<1$,
\begin{equation}\label{eq:mixing-rate}
 \Cov_{(1-\eta)*\mu^S}(X)
 \preceq c(\eta)\diag\bigl(\Ex_{(1-\eta)*\mu^S}X_e\bigr).
\end{equation}
If the integral below is finite, then
\begin{equation}\label{eq:mixing-comparison}
 \gap(P_{\mathrm{GD}})
 \ge\frac{1}{(1+r_{\mathrm{ins}})m}
 \exp\left(-\int_0^1\frac{c(\eta)-1}{1-\eta}\,d\eta\right).
\end{equation}
\end{theorem}
For matching models with small fugacity, the following covariance bound is provided.
\begin{theorem}[Covariance bound at small fugacity~{\cite[Theorem~1.14]{CCCYZ25}}]\label{thm:small-fugacity-covariance}
Let $\mu_s$ be a $b$-matching law with vertex-dependent capacities at fugacity $0<s<1$. For every feasible occupied set $S$,
\begin{equation}\label{eq:small-activity-cov}
 \Cov_{\mu_s^S}(X)
 \preceq\frac{1}{1-s}\diag\bigl(\Ex_{\mu_s^S}X_e\bigr).
\end{equation}
\end{theorem}

\section{Spectral independence of occupied degrees}\label{sec:degrees}

Degree coupling independence gives a spectral-gap bound for block dynamics. We use this bound to prove \cref{lem:overview-degree} and the revealing estimate needed in \cref{sec:surrogate}.

\begin{definition}[Coupling independence in the degree metric]\label{def:degree-coupling}
Let $\Omega=\prod_{v\in V}\Omega_v$, where each $\Omega_v$ is a finite set of nonnegative integers, and let $\nu$ be a law on $\Omega$. Equip these degree vectors with the \emph{degree metric}
\[
 d_{\mathrm{deg}}(D,D')=\sum_v|D_v-D'_v|,
\]
which measures the total change in occupied degrees across vertices. For every $\Lambda\subseteq V$ and $\tau\in\prod_{v\in\Lambda}\Omega_v$, including impossible pinnings, choose a probability law $\nu^\tau$ on $\Omega$ supported on $D_\Lambda=\tau$. Whenever the pinning is feasible, require $\nu^\tau=\nu(\,\cdot\mid D_\Lambda=\tau)$.

For $C\ge1$, we say that $\nu$ has \emph{degree coupling independence} with constant $C$ if these laws can be chosen so that, for every such $\Lambda,\tau$, every $v\notin\Lambda$, and distinct $k,\ell\in\Omega_v$, there is a coupling of
\[
 D^{(k)}\sim\nu^{\tau,v=k},\qquad
 D^{(\ell)}\sim\nu^{\tau,v=\ell}
\]
such that
\begin{equation}\label{eq:coupling-distance}
 \Ex\bigl[d_{\mathrm{deg}}(D^{(k)},D^{(\ell)})\bigr]\le C|k-\ell|.
\end{equation}
\end{definition}

\begin{remark}
  In the definition of coupling independence, we work with an extension of the conditional distributions to all degree pinnings rather than feasible pinnings.
\end{remark}

The factor $|k-\ell|$ measures the size of the change at $v$. Since the contribution at $v$ is exactly $|k-\ell|$, the expected discrepancy outside $v$ is at most $(C-1)|k-\ell|$.

The following degree coupling bound follows from the recursive coupling in~\cite{CG24}.
\begin{theorem}[Degree coupling independence~{\cite{CG24}}]\label{obs:coupling}
Consider a well-defined Holant measure with positive activities, where each vertex signature $f_v$ is either log-concave with $f_v(0)>0$ or an exact-degree signature (supported at a single degree). Its degree-vector law has degree coupling independence with constant $2$ on $\Omega=\prod_v\operatorname{supp}f_v$.
\end{theorem}

\begin{proof}
We first construct a coupling for two feasible adjacent root degrees. We then add pendant edges to make all degree pinnings feasible and pass to a limit.

Fix a feasible degree pinning $\tau$ and a vertex $v$ outside its domain. Suppose that both extensions $D_v=k$ and $D_v=k+1$ are feasible. We couple the corresponding conditional laws as follows. Consecutive root pinnings correspond to the signatures $\one_{\{j=k+1\}}$ and $\one_{\{j=k\}}$, which differ by a downward shift. This is the case handled by the coupling algorithm in \cite[Algorithm~1 and Lemma~18]{CG24}; its validity uses log-concavity and allows exact-degree signatures.

The degree bound is preserved at each recursive step. If the two samples assign the same value to an edge, no discrepancy is added. If they disagree, the recursion moves the active defect to the other endpoint and reverses the order of the two samples. Thus successive disagreements contribute opposite signs at their shared vertex: the next disagreement cancels the previous degree change there and moves it to the new endpoint. Inductively, among the edges already processed, only the starting vertex and the current endpoint can have a degree discrepancy, with total magnitude at most two. This remains true even when the recursion revisits a vertex. Thus the coupling has degree distance at most two for these adjacent root degrees.

We now construct the conditional laws for all pinnings in $\Omega$. Attach $b_v=\max\operatorname{supp}f_v$ private leaves to each original vertex $v$, give each new edge activity $\varepsilon>0$ and each new leaf signature $(1,1)$, and extend the original signatures by zero outside their original domains. The degree law $\nu_\varepsilon$ on the original vertices has full support on $\Omega$: every vector in $\Omega$ can be realized using only the pendant edges. Thus every original-coordinate pinning is feasible, and each unpinned root has its entire interval support available. Apply the preceding coupling to each pair of adjacent root degrees after any such pinning and project away the leaf coordinates, which only decreases the degree distance. For arbitrary root degrees $k<\ell$ in the support, every intermediate degree is now feasible. Gluing the couplings for consecutive degrees and applying the triangle inequality gives expected degree distance at most $2(\ell-k)$.

As $\varepsilon\downarrow0$, $\nu_\varepsilon$ converges to $\nu$. Since there are finitely many pinnings and coupling pairs, take one subsequence along which all their conditional laws and couplings converge. The limiting laws $\nu^\tau$ are supported on their prescribed pinnings and agree with ordinary conditioning whenever the pinning is feasible under $\nu$. All coupling bounds pass to the limit, proving \eqref{eq:coupling-distance} for pinnings throughout $\Omega$.
\end{proof}

It is well known that coupling independence implies spectral independence; see \cite[Section~2.1]{CF24}. We apply the same coupling argument directly to dynamics that retain a random subset of degrees.

\begin{theorem}[Block dynamics for occupied degrees]\label{thm:bessel}
Let $\nu$ be a law on $\Omega=\prod_{v\in V}\Omega_v$ with degree coupling independence with constant $2$. Let $R\subseteq V$ be a random subset, independent of $D$, such that
\[
 r:=\max_{v\in V}\Pr(v\in R)<\frac12.
\]
For every real function $F$ of the degree vector,
\begin{equation}\label{eq:reveal-gap}
 \Ex_R\Ex_\nu\bigl[\Var{\nu}{F(D)\mid D_R}\bigr]
 \ge(1-2r)\Var{\nu}{F(D)}.
\end{equation}
In particular, the distribution $\nu$ satisfies
\begin{equation}\label{eq:bessel}
 \sum_v\Var{\nu}{\Ex_\nu[F(D)\mid D_v]}\le2\Var{\nu}{F(D)}.
\end{equation}
\end{theorem}

\begin{remark}
Consider the block dynamics that, at each step, draw $R$, keep $D_R$ fixed, and jointly resample the block $D_{V\setminus R}$ from its conditional distribution given $D_R$. The left side of \eqref{eq:reveal-gap} is exactly the Dirichlet form of this chain, so the theorem states that its spectral gap is at least $1-2r$. See~\cite{BCCPSV21} for a detailed explanation.
\end{remark}

For completeness, we give the standard proof that coupling independence implies spectral independence, adapted to the degree metric. We use block dynamics as in \cite[Lemma~33]{CF24} and deduce a spectral-gap bound from coupling contraction as in \cite[Theorem~43]{AJKPV24}.
We use M.~F.~Chen's contraction lemma in the form stated in~\cite[Theorem~13.1]{LP17}.

\begin{lemma}[M.~F.~Chen~{\cite{Chen98,LP17}}]\label{lem:chen-contraction}
Let $(\Omega,d)$ be a finite metric space and let $P$ be the transition matrix of a Markov chain on $\Omega$. Suppose that there exists $\theta\in[0,1)$ such that, for every $x,y\in\Omega$, there is a coupling $(X_1,Y_1)$ of $P(x,\cdot)$ and $P(y,\cdot)$ satisfying
\[
 \Ex_{x,y}[d(X_1,Y_1)]\le\theta d(x,y).
\]
If $\lambda\ne1$ is an eigenvalue of $P$, then $|\lambda|\le\theta$. In particular, the absolute spectral gap satisfies
\[
 \gamma_*\ge1-\theta.
\]
\end{lemma}

\begin{proof}[Proof of \cref{thm:bessel}]
Define $P$ by drawing $R$, retaining $D_R$, and resampling all other degrees conditionally. This chain is reversible with respect to $\nu$, with Dirichlet form
\[
 \Diri{P}{F,F}=\Ex_R\Ex_\nu[\Var{\nu}{F\mid D_R}].
\]
We will couple one step so that the expected degree distance contracts by a factor at most $2r$.

Fix $x,y\in\operatorname{supp}\nu$ and draw the same $R$ in both chains. Change the pinning $x_R$ to $y_R$ one coordinate at a time. The laws at every intermediate pinning are supplied by \cref{def:degree-coupling}, even when the pinning is impossible under $\nu$. Coupling each consecutive pair and gluing these couplings gives a coupling of $\nu^{x_R}$ and $\nu^{y_R}$ with expected degree distance at most $2\sum_{v\in R}|x_v-y_v|$. Both endpoint pinnings are feasible, so these are the required conditional resampling laws. Averaging over $R$ yields
\[
 \Ex[d_{\mathrm{deg}}(X',Y')]
 \le2\sum_{v\in V}\Pr(v\in R)|x_v-y_v|
 \le2r\,d_{\mathrm{deg}}(x,y).
\]
By \cref{lem:chen-contraction}, $P$ has spectral gap at least $1-2r$, proving \eqref{eq:reveal-gap}.

Finally, when $n=|V|>2$, take $R$ to be a uniformly chosen singleton, so $r=1/n$. By total variance, \eqref{eq:reveal-gap} becomes \eqref{eq:bessel}. For $n\le2$, \eqref{eq:bessel} follows directly from $\Var{}{\Ex[F\mid D_v]}\le\Var{}{F}$. The same argument applies after every feasible degree pinning, since degree coupling independence is preserved by conditioning. By the variance characterization of spectral independence~\cite[Lemma~21]{AJKPV24}, \eqref{eq:bessel} under these pinnings gives $1$-spectral independence in the vertex coordinates.
\end{proof}

Now, we are ready to prove \Cref{lem:overview-degree}.
\begin{proof}[Proof of \cref{lem:overview-degree}]
The dual form of \eqref{eq:bessel} gives, for arbitrary $g_v$,
\begin{equation}\label{eq:degree-additive-spectral}
 \Var{\nu}{\sum_vg_v(D_v)}\le2\sum_v\Var{\nu}{g_v(D_v)}.
\end{equation}
To obtain \eqref{eq:overview-degree}, recall that $q_v=\Pr(D_v<b_v)$ is the probability that $v$ is not saturated, and take $g_v(k)=h_vk$. Then
\[
 \Var{}{\sum_vh_vD_v}\le2\sum_vh_v^2\Var{}{D_v}\le2\sum_vh_v^2\Ex[(D_v-b_v)^2]\le2b^2\sum_vq_vh_v^2.
\]
Here variance is bounded by the mean squared distance from any constant, and
$(D_v-b_v)^2\le b^2\one_{\{D_v<b_v\}}$ since $0\le D_v\le b_v\le b$.
\end{proof}

We also record the following corollary of \cref{thm:bessel}.
\begin{corollary}[A quadratic revealing estimate]\label{lem:quadratic-reveal}
Let $\nu$ satisfy the assumptions of \cref{thm:bessel}, and let $G=(V,E)$ be a simple graph. For each edge $e=uv$, let $K_e=K_e(D_u,D_v)$ be any real function of its endpoint degrees. For every real edge vector $a$,
\begin{equation}\label{eq:quadratic-reveal}
 \Var{}{\sum_ea_eK_e}
 \le6\sum_v\Ex\left[\left(\sum_{e\ni v}a_eK_e\right)^2\right]
   +12\sum_ea_e^2\Ex[K_e^2].
\end{equation}
\end{corollary}
\begin{proof}
Fix $t\in(0,1/2)$, and let $R\subseteq V$ include each vertex independently with probability $t$, independently of $D$. Applying \eqref{eq:reveal-gap} from \cref{thm:bessel} with $F(D)=\sum_ea_eK_e(D_u,D_v)$ gives
\[
 (1-2t)\Var{}{\sum_ea_eK_e}
 \le\Ex_R\Ex_D\left[\Var{}{\sum_ea_eK_e\mid D_R}\right].
\]
We bound the right-hand side using the dependence of each $K_e$ on its two endpoint degrees. For each fixed $R$, the sum $\sum_{e=uv}a_eK_e\one_{\{u,v\in R\}}$ is determined by $D_R$, since every surviving term has both endpoints in $R$. The mean-square optimality of conditional expectation therefore implies
\[
 \Ex_R\Ex_D\left[\Var{}{\sum_ea_eK_e\mid D_R}\right]
 \le\Ex_R\Ex_D\left[
   \left(\sum_ea_eK_e-t^{-2}\sum_{e=uv}a_eK_e\one_{\{u,v\in R\}}\right)^2
 \right].
\]
Conditional on $D$, the coefficients $a_eK_e$ are fixed and
\[
 \Ex_R\left[\sum_{e=uv}a_eK_e\one_{\{u,v\in R\}}\;\middle|\;D\right]
 =t^2\sum_ea_eK_e.
\]
Thus, exchanging the order of the two expectations gives
\[
\begin{aligned}
 &\Ex_R\Ex_D\left[
   \left(\sum_ea_eK_e-t^{-2}\sum_{e=uv}a_eK_e\one_{\{u,v\in R\}}\right)^2
 \right] =t^{-4}\Ex_D\left[
   \Var{R}{\sum_{e=uv}a_eK_e\one_{\{u,v\in R\}}\mid D}
 \right].
\end{aligned}
\]
To compute this variance, condition on $D$ and expand the centered products:
\[
\begin{aligned}
 \sum_{e=uv}a_eK_e\one_{\{u,v\in R\}}-t^2\sum_ea_eK_e
 &=t\sum_v\left(\sum_{e\ni v}a_eK_e\right)\bigl(\one_{\{v\in R\}}-t\bigr)\\
 &+\sum_{e=uv}a_eK_e\bigl(\one_{\{u\in R\}}-t\bigr)\bigl(\one_{\{v\in R\}}-t\bigr).
\end{aligned}
\]
The centered indicators $\one_{\{v\in R\}}-t$ are independent, with mean zero and second moment $t(1-t)$. Since $G$ is simple, distinct linear and quadratic monomials involve different vertex sets. Their product therefore contains a centered indicator appearing exactly once, so independence makes its expectation over $R$ zero, even when the edges share a vertex. Thus all cross terms vanish, giving
\[
\begin{aligned}
 \Var{R}{\sum_{e=uv}a_eK_e\one_{\{u,v\in R\}}\mid D}
 &=t^3(1-t)\sum_v\left(\sum_{e\ni v}a_eK_e\right)^2+t^2(1-t)^2\sum_ea_e^2K_e^2.
\end{aligned}
\]
Collecting the shared linear contributions into vertex sums and eliminating the remaining cross terms avoids introducing a maximum-degree factor.

Combining the preceding estimates yields
\[
\begin{aligned}
 (1-2t)\Var{}{\sum_ea_eK_e}
 &\le\Ex_R\Ex_D\left[\Var{}{\sum_ea_eK_e\mid D_R}\right]\\
 &\le\frac{1-t}{t}
       \sum_v\Ex\left[\left(\sum_{e\ni v}a_eK_e\right)^2\right]
   +\frac{(1-t)^2}{t^2}\sum_ea_e^2\Ex[K_e^2].
\end{aligned}
\]
Taking $t=1/3$ and multiplying by $3$ proves \eqref{eq:quadratic-reveal}.
\end{proof}

\section{Edge-to-vertex variance bound}\label{sec:surrogate}

In this section we prove \cref{lem:overview-transfer}. Work in a fixed feasible residual instance under the assumptions of \cref{thm:holant}, omitting deterministic edges and vertices with no live incident edges. Throughout this section, $b_v$ is the residual support bound, $p_e=\Pr(X_e=1)$, and $q_v=\Pr(D_v<b_v)>0$.

We first outline the proof ideas of~\Cref{lem:overview-transfer}. If we could take $K_e=X_e$ in \cref{lem:quadratic-reveal}, the result would follow immediately. Indeed, Cauchy--Schwarz over the occupied edges at $v$ gives
\[
 \left(\sum_{e\ni v}a_eX_e\right)^2
 \le D_v\sum_{e\ni v}a_e^2X_e
 \le b\sum_{e\ni v}a_e^2X_e.
\]
Taking expectations and summing over vertices counts each edge twice. Since $X_e^2=X_e$ and $\Ex X_e=p_e$, the corollary would therefore give
\[
\begin{aligned}
 \Var{}{\sum_ea_eX_e}
 &\le6\sum_v\Ex\left[\left(\sum_{e\ni v}a_eX_e\right)^2\right]
       +12\sum_ep_ea_e^2\le(12b+12)\sum_ep_ea_e^2.
\end{aligned}
\]
This would prove \cref{lem:overview-transfer} with $C_1=12(b+1)$ and $C_2=0$; the required dependence of $C_1$ holds because $T\ge1$.

For an edge $e=uv$, the endpoint degrees $D_u,D_v$ do not in general determine $X_e$, so \cref{lem:quadratic-reveal} cannot be applied directly to the edge indicators. We instead construct a surrogate $K_e=\phi_e(D_u,D_v)$ that depends only on the endpoint degrees and approximates $X_e$.
Specifically,
we require $K_e$ and $X_e$ to have the same conditional mean given all other edge indicators:
\begin{align}\label{eq:conditional-eq}
 \Ex[K_e\mid X_{-e}]=\Ex[X_e\mid X_{-e}],
\end{align}
where $X_{-e}$ denotes all edge indicators other than $X_e$. This gives $\Ex K_e=p_e$ and $\Ex[K_eX_f]=\Ex[X_eX_f]$ for $f\ne e$.

Whenever these moment identities holds, put $F=\sum_ea_eX_e$ and $F_K=\sum_ea_eK_e$. The following identity transfers the variance bound with only an individual-edge correction:
\begin{equation}\label{eq:comparison-identity}
 \Var{}{F}=\Var{}{F_K}
 +2\sum_ea_e^2\bigl(\Var{}{X_e}-\Cov(X_e,K_e)\bigr)-\Var{}{F-F_K}.
\end{equation}
Indeed, for distinct edges $e\ne f$, the moment identities give $\Cov(X_e,K_f)=\Cov(X_e,X_f)$. Thus the terms indexed by $e\ne f$ cancel when we subtract $\Cov(F,F_K)$ from $\Var{}{F}$, yielding
\[
 \Var{}{F}-\Cov(F,F_K)
 =\sum_ea_e^2\bigl(\Var{}{X_e}-\Cov(X_e,K_e)\bigr).
\]
Expanding $\Var{}{F-F_K}=\Var{}{F}+\Var{}{F_K}-2\Cov(F,F_K)$ then proves \eqref{eq:comparison-identity}.

We first illustrate this approach for ordinary matchings, where a simple surrogate allows us to prove \cref{lem:overview-transfer} directly using \cref{lem:quadratic-reveal} and \eqref{eq:comparison-identity}.

\begin{proof}[Proof of \cref{lem:overview-transfer} for matchings]
For matchings at activity $\lambda>0$, we prove \cref{lem:overview-transfer} with $C_1=2+36\lambda$ and $C_2=6$. Set
\[
 K_{uv}=\lambda\one_{\{D_u=D_v=0\}},\qquad q_v=\Pr(D_v=0).
\]
Adding or removing $e$ gives $\Ex K_e=p_e$ and $\Ex[K_eX_f]=\Ex[X_eX_f]$ for $f\ne e$. Moreover, $X_eK_e=0$ and $K_e^2=\lambda K_e$.

The star sum vanishes when $D_v=1$ and has mean $\sum_{e\ni v}p_ea_e$. Its conditional mean given $D_v=0$ is therefore $q_v^{-1}\sum_{e\ni v}p_ea_e$. The law of total variance gives the equality below:
\begin{align*}
 \Var{}{\sum_{e\ni v}a_eK_e}
 &=q_v\Var{}{\sum_{e\ni v}a_eK_e\mid D_v=0}
   +\frac{1-q_v}{q_v}\left(\sum_{e\ni v}p_ea_e\right)^2\\
 &\le2\lambda\sum_{e\ni v}p_ea_e^2
   +\frac{1-q_v}{q_v}\left(\sum_{e\ni v}p_ea_e\right)^2.
\end{align*}
For the inequality, apply \cref{lem:overview-degree} conditional on $D_v=0$: the law is the matching law on $G-v$, and each $K_{vu}=\lambda\one_{\{D_u=0\}}$ depends on a distinct neighboring degree. Use also $q_v\Ex[K_e^2\mid D_v=0]=\Ex K_e^2=\lambda p_e$.
By \eqref{eq:comparison-identity} and $X_eK_e=0$, we obtain the first inequality below. For the second, apply \cref{lem:quadratic-reveal} to $K_e-p_e$, use $\Var{}{K_e}\le\lambda p_e$, and count each edge in two stars:
\begin{align*}
 \Var{}{F}
 &\le\Var{}{F_K}+2\sum_ep_ea_e^2
 \le(2+36\lambda)\sum_ep_ea_e^2
   +6\sum_v\frac{1-q_v}{q_v}\left(\sum_{e\ni v}p_ea_e\right)^2.
\end{align*}
Since $1-q_v\le1$, this proves \cref{lem:overview-transfer}. The argument in \cref{sec:overview} then gives $O((1+\lambda)^2)$ spectral independence.
\end{proof}

\subsection{Harmonic surrogates}\label{sec:harmonic-surrogates}

For general log-concave signatures, we construct the surrogate from the conditional-mean requirement~\eqref{eq:conditional-eq}. Given $X_{-e}$, the absent and present states of an available edge have endpoint degrees $(i,j)$ and $(i+1,j+1)$, with relative weights $1:\theta_e(i,j)$. Thus the requirement becomes $\phi_e(i,j)+\theta_e(i,j)\phi_e(i+1,j+1)=\theta_e(i,j)$. Setting $\phi_e=0$ when either endpoint is saturated gives the following backward recursion.

\begin{definition}[Harmonic surrogates]\label{def:harmonic-surrogates}
For each edge $e=uv$, define $\phi_e(i,j)$ for $0\le i\le b_u$ and $0\le j\le b_v$ by
\begin{equation}\label{eq:harmonic}
 \begin{split}
 \phi_e(i,j)&=0\quad\text{if }i=b_u\text{ or }j=b_v,\\
 \phi_e(i,j)&=\theta_e(i,j)\{1-\phi_e(i+1,j+1)\}
       \quad(i<b_u,\ j<b_v),
 \end{split}
\end{equation}
and set $K_e=\phi_e(D_u,D_v)$.
\end{definition}

\begin{lemma}[Matching conditional means]\label{lem:moment}
For every edge $e$,
\begin{equation}\label{eq:moment}
 \Ex[K_e\mid X_{-e}]=\Ex[X_e\mid X_{-e}].
\end{equation}
Consequently, $\Ex K_e=p_e$ and $\Ex[K_eX_f]=\Ex[X_eX_f]$ for $f\ne e$. Also, for either endpoint $v$ of $e=uv$,
\begin{equation}\label{eq:insertion}
 p_e=\Ex[(1-X_e)\theta_e(D_u,D_v)]\le\Theta q_v.
\end{equation}
\end{lemma}
\begin{proof}
Condition on $X_{-e}$. If $e$ is unavailable, an endpoint is saturated and both $K_e$ and $X_e$ vanish. Otherwise let $i,j$ be the absent endpoint degrees and $\theta=\theta_e(i,j)$. The two conditional weights are $1:\theta$, so
\[
 \Ex[K_e\mid X_{-e}]
 =\frac{\phi_e(i,j)+\theta\phi_e(i+1,j+1)}{1+\theta}
 =\frac{\theta}{1+\theta}
 =\Ex[X_e\mid X_{-e}].
\]
Multiply by $1$ or $X_f$ and average to obtain the moment identities. The same two-state calculation gives the equality in \eqref{eq:insertion}; its integrand is at most $\Theta\one_{\{D_v<b_v\}}$.
\end{proof}

\begin{lemma}[Uniform surrogate bound]\label{lem:surrogate-size}
For $L=L_b(\Theta)$ from \eqref{eq:main-constant}, $|\phi_e(i,j)|\le L$. If $\Theta\le1$, then $0\le\phi_e(i,j)\le\Theta$.
\end{lemma}
\begin{proof}
Let $\ell=\min\{b_u-i,b_v-j\}\le b$. If $\Theta\le1$, induction in \eqref{eq:harmonic}, starting from the zero boundary, gives $0\le\phi_e\le\Theta$.

For $\Theta\ge1$, use the stronger invariant
$1-\Theta^\ell\le\phi_e(i,j)\le\Theta^\ell$, which holds at $\ell=0$. The inductive hypothesis puts $1-\phi_e(i+1,j+1)$ in $[1-\Theta^{\ell-1},\Theta^{\ell-1}]$. Since its lower endpoint is nonpositive and $0\le\theta_e(i,j)\le\Theta$,
\[
 1-\Theta^\ell
 \le\Theta(1-\Theta^{\ell-1})
 \le\phi_e(i,j)
 \le\Theta^\ell.
\]
Thus $|\phi_e(i,j)|\le\Theta^\ell\le\Theta^b=L$.
\end{proof}

\subsection{From star variances to edge variance}\label{sec:stars}\label{sec:edge-comparison}

The star sum has mean $\sum_{e\ni v}p_ea_e$ and vanishes when $v$ is saturated. The law of total variance on $A_v=\{D_v<b_v\}$ therefore gives the collective term below exactly. We bound the variance within $A_v$ by comparing with the simpler boundary values of the recursion.

\begin{samepage}
\begin{lemma}[Centered-star bound]\label{lem:source}
Put
\begin{equation}\label{eq:star-constant}
 A_\star=4\Theta+72(b-1)T^2,
\end{equation}
where $T$ is defined in \eqref{eq:main-constant}. For every vertex $v$ and every real edge vector $a$,
\begin{equation}\label{eq:source-local}
 \Var{}{\sum_{e\ni v}a_eK_e}
 \le A_\star\sum_{e\ni v}p_ea_e^2
       +\frac{1-q_v}{q_v}\left(\sum_{e\ni v}p_ea_e\right)^2.
\end{equation}
\end{lemma}
\end{samepage}
We defer the proof of \cref{lem:source} to \cref{sec:star-moments} and first deduce \cref{lem:overview-transfer}.

\begin{proof}[Proof of \cref{lem:overview-transfer}]
First take a single nonzero coefficient in \eqref{eq:source-local}. By \eqref{eq:insertion},
\begin{equation}\label{eq:surrogate-variance}
 \Var{}{K_e}\le A_\star p_e+\frac{1-q_v}{q_v}p_e^2
 \le(A_\star+\Theta)p_e\qquad(e\ni v).
\end{equation}
Thus the individual variance needs no separate proof.

By \cref{lem:moment}, the moment identities needed for \eqref{eq:comparison-identity} hold. Drop its final nonpositive term and use
$2|\Cov(X_e,K_e)|\le\Var{}{X_e}+\Var{}{K_e}$ to obtain
\[
 \Var{}{F}\le\Var{}{F_K}+3\sum_ea_e^2\Var{}{X_e}+\sum_ea_e^2\Var{}{K_e}.
\]
Apply \cref{lem:quadratic-reveal} to the centered functions $K_e-p_e$, so its second moments become variances. Then use \eqref{eq:source-local}, \eqref{eq:surrogate-variance}, and $\Var{}{X_e}\le p_e$. Each edge occurs in two stars, so
\begin{align*}
 \Var{}{F}
 &\le6\sum_v\Var{}{\sum_{e\ni v}a_eK_e}
       +13\sum_ea_e^2\Var{}{K_e}+3\sum_ea_e^2\Var{}{X_e}\\
 &\le(25A_\star+13\Theta+3)\sum_ep_ea_e^2
       +6\sum_v\frac{1-q_v}{q_v}\left(\sum_{e\ni v}p_ea_e\right)^2.
\end{align*}
Using $1-q_v\le1$ and substituting \eqref{eq:star-constant} gives
\begin{equation}\label{eq:edge-source}
 \begin{aligned}
 \Var{}{\sum_ea_eX_e}
 &\le\bigl[3+113\Theta+1800(b-1)T^2\bigr]\sum_ep_ea_e^2
 +6\sum_v\frac{\left(\sum_{e\ni v}p_ea_e\right)^2}{q_v}.
 \end{aligned}
\end{equation}
This proves \cref{lem:overview-transfer} with $C_1=3+113\Theta+1800(b-1)T^2$ and $C_2=6$, preserving the claimed parameter dependence.
\end{proof}

\subsection{Proof of the centered-star bound}\label{sec:star-moments}

\begin{proof}[Proof of \cref{lem:source}]
Fix $v$ and write $A_v=\{D_v<b_v\}$, with $\Pr(A_v)=q_v$. Every $K_e$ in the star vanishes when $A_v$ does not occur, and \cref{lem:moment} gives $\Ex K_e=p_e$. Thus the conditional mean is
\[
 \Ex\left[\sum_{e\ni v}a_eK_e\,\middle|\,\one_{A_v}\right]
 =\frac{\one_{A_v}}{q_v}\sum_{e\ni v}p_ea_e.
\]
Apply the law of total variance by conditioning on $\one_{A_v}$. The conditional variance is zero when $A_v$ does not occur, and $\Var{}{\one_{A_v}}=q_v(1-q_v)$, so
\begin{equation}\label{eq:star-centering}
 \begin{aligned}
 \Var{}{\sum_{e\ni v}a_eK_e}
 &=q_v\Var{}{\sum_{e\ni v}a_eK_e\mid A_v}
     +\left(\frac{\sum_{e\ni v}p_ea_e}{q_v}\right)^2\Var{}{\one_{A_v}}\\
 &=q_v\Var{}{\sum_{e\ni v}a_eK_e\mid A_v}
     +\frac{1-q_v}{q_v}\left(\sum_{e\ni v}p_ea_e\right)^2.
 \end{aligned}
\end{equation}
It remains to bound the conditional variance term $\Var{}{\sum_{e\ni v}a_eK_e\mid A_v}$.
The difficulty here is that all the $K_e$ share the same variable $D_v$.
For $e=vu$, we introduce an approximation $\widehat{K}_e$ that depends only on $D_u$ after conditioning on $A_v$:
\[
 \widehat{K}_e:=\one_{A_v}\theta_e(b_v-1,D_u) = \one_{A_v}\lambda_er_v(b_v-1)r_u(D_u),\qquad
 R_e:=K_e-\widehat{K}_e,
\]
where $R_e$ is the approximation error.
When $D_v=b_v-1$, the recursion has one step left and gives $K_e=\widehat{K}_e$. When $D_v=b_v$, both terms defining $R_e$ vanish. Hence
\begin{equation}\label{eq:remainder-support}
 R_e=0\ \text{unless }D_v\le b_v-2,
 \qquad 0\le \widehat{K}_e\le\Theta,\qquad |R_e|\le2L.
\end{equation}
We have $K_e=\widehat{K}_e+R_e$. Using
$\Var{}{U+W}\le2\Var{}{U}+2\Var{}{W}$ and bounding the last variance by a second moment gives
\begin{equation}\label{eq:star-split}
 \begin{aligned}
 q_v\Var{}{\sum_{e\ni v}a_eK_e\mid A_v}
 &\le2q_v\Var{}{\sum_{e=vu}a_e\widehat{K}_e\mid A_v}
 +2\Ex\left[\left(\sum_{e\ni v}a_eR_e\right)^2\right].
 \end{aligned}
\end{equation}
We bound these two fluctuation terms separately.

\medskip\noindent\emph{Boundary fluctuation.}
Apply \cref{lem:overview-degree} under the conditioning on $A_v$: upper-truncating $f_v$ preserves log-concavity, and each $\widehat{K}_e$ depends only on $D_u$, with distinct neighbors $u$ since $G$ is simple. As $\widehat{K}_e=0$ when $A_v$ does not occur, this gives the first inequality below:
\begin{equation}\label{eq:boundary-fluctuation}
 \begin{aligned}
 q_v\Var{}{\sum_{e=vu}a_e\widehat{K}_e\mid A_v}
 &\le2\sum_{e\ni v}a_e^2\Ex[\widehat{K}_e^2]
 \le2\Theta^2(1+\gamma^{-1})\sum_{e\ni v}p_ea_e^2.
 \end{aligned}
\end{equation}
For the second inequality, first split according to $X_e$ and use $0\le\widehat{K}_e\le\Theta$:
\begin{align*}
 \Ex[\widehat{K}_e^2]
 &=\Ex[X_e\widehat{K}_e^2]+\Ex[(1-X_e)\widehat{K}_e^2]
 \le\Theta^2p_e+\Ex[(1-X_e)\widehat{K}_e^2].
\end{align*}
Note that, whenever $\widehat{K}_e>0$, both endpoints are unsaturated, so $\theta_e(D_v,D_u)\ge\gamma$. Thus $\widehat{K}_e^2\le\Theta^2\le(\Theta^2/\gamma)\theta_e(D_v,D_u)$; the bound $\widehat{K}_e^2\le(\Theta^2/\gamma)\theta_e(D_v,D_u)$ also holds when $\widehat{K}_e=0$. By \eqref{eq:insertion},
\[
 \Ex[(1-X_e)\widehat{K}_e^2]
 \le\frac{\Theta^2}{\gamma}\Ex[(1-X_e)\theta_e(D_v,D_u)]
 =\frac{\Theta^2}{\gamma}p_e.
\]
Substituting yields $\Ex[\widehat{K}_e^2]\le\Theta^2(1+\gamma^{-1})p_e$, as required.

\medskip\noindent\emph{Approximation error.}
Recall that $L=\max\{\Theta,\Theta^b\}$ and $T=L/\min\{1,\gamma\}$. We will prove
\begin{equation}\label{eq:remainder-second}
 \Ex\left[\left(\sum_{e\ni v}a_eR_e\right)^2\right]
 \le32(b-1)T^2\sum_{e\ni v}p_ea_e^2.
\end{equation}
By \eqref{eq:remainder-support}, if $b_v=1$, all $R_e$ vanish. Otherwise, condition on the edges outside the star of $v$. Let $\mathcal A$ be the incident edges whose other endpoints are not saturated by this exterior, and put $d=|\mathcal A|$. Every subset $J\subseteq\mathcal A$ with $|J|\le b_v$ is feasible; edges outside $\mathcal A$ have $K_e=\widehat{K}_e=R_e=0$. All probabilities in the next calculation are conditional on this exterior, and $\mu(J)$ denotes the conditional probability of $J$. If $d=0$, there is nothing to prove.

By \cref{lem:surrogate-size}, $|K_e|\le L$, so $|R_e|\le2L$. By \eqref{eq:remainder-support} and Cauchy--Schwarz,
\[
 \Ex\left[\left(\sum_{e\ni v}a_eR_e\right)^2\right]
 \le4L^2d\Pr(D_v\le b_v-2)\sum_{e\in\mathcal A}a_e^2.
\]
Thus it suffices to prove, for every $e\in\mathcal A$,
\begin{equation}\label{eq:two-vacancies-probability}
 d\Pr(D_v\le b_v-2)
 \le\frac{8(b_v-1)}{\min\{1,\gamma\}^2}\Pr(X_e=1).
\end{equation}
Indeed, applying this estimate to each summand replaces the factor $d\Pr(D_v\le b_v-2)$ by a multiple of $\Pr(X_e=1)$. Averaging over the exterior then gives \eqref{eq:remainder-second}, using $b_v\le b$ and $T=L/\min\{1,\gamma\}$.

To prove \eqref{eq:two-vacancies-probability}, fix $e\in\mathcal A$ and count pairs $(J,f)$ with $|J|\le b_v-2$ and $f\in\mathcal A$ by the map
\[
 (J,f)\longmapsto H=J\cup\{e,f\}.
\]
At most two edges are inserted. The two vacancies at $v$ and the distinct other endpoints make both insertions feasible. Each has weight ratio at least $\gamma$, so
\[
 \mu(J)\le\frac{\mu(H)}{\min\{1,\gamma\}^2}.
\]
For each image $H$, we have $f\in H$ and
$H\setminus\{e,f\}\subseteq J\subseteq H$. There are at most four choices of $J$ for each $f$, hence at most $4|H|\le4b_v\le8(b_v-1)$ preimages. Summing the weight comparison over all pairs $(J,f)$ proves \eqref{eq:two-vacancies-probability}: the total source weight is $d\Pr(D_v\le b_v-2)$, and every image contains $e$.

\medskip\noindent\emph{Wrapping up.}
Substituting the boundary and error estimates into \eqref{eq:star-split}, we obtain
\[
 q_v\Var{}{\sum_{e\ni v}a_eK_e\mid A_v}
 \le A_\star\sum_{e\ni v}p_ea_e^2.
\]
To check the coefficient, if $b\ge2$, then $L\ge\Theta$ implies
\[
 \Theta^2(1+\gamma^{-1})\le2T^2
 \le2(b-1)T^2.
\]
Thus \eqref{eq:boundary-fluctuation} and \eqref{eq:remainder-second} give a coefficient of at most $72(b-1)T^2\le A_\star$.
If $b=1$, the error vanishes and $X_e=0$, $\widehat{K}_e=\theta_e(D_v,D_u)$ on $A_v$. Thus \eqref{eq:insertion} improves the boundary estimate to
$\Ex[\widehat{K}_e^2]\le\Theta p_e$, giving the coefficient $4\Theta=A_\star$.
Substituting into \eqref{eq:star-centering} proves \eqref{eq:source-local}.
\end{proof}

\section*{Acknowledgements}
We thank Zongchen Chen, Heng Guo, and Yitong Yin for helpful discussions.

\begingroup
\raggedright
\bibliographystyle{alpha}
\bibliography{refs}
\endgroup
\clearpage
\appendix
\section{An alternative proof of the main theorem}\label{app:alternative-proof}

Fix a feasible pinning $\tau$.
Work in this residual instance, writing $p_e=\Ex X_e$ and $q_v=\Pr(D_v<b_v)$ and omitting deterministic edges and vertices with no live incident edges. The one-state case is immediate; otherwise, consider the optimization problem
\[
 \rho=\max_{a\ne0}\frac{\Var{}{\sum_ea_eX_e}}{\sum_ep_ea_e^2}.
\]
This mean-normalized quantity is the local covariance quantity underlying spectral stability for field dynamics~\cite{CE22,CCYZ25,CCCYZ25}.

Choose a maximizer $a^\star$ with $\sum_ep_e(a_e^\star)^2=1$, so $\Var{}{\sum_ea_e^\star X_e}=\rho$.
Since a live edge is nonconstant, $\rho>0$.
We claim that
\begin{align} \label{eq:proof-thm:holant-claim}
\rho \sum_v\frac{\left(\sum_{e\ni v}p_e a^\star_e\right)^2}{q_v} \leq 2b^2.
\end{align}

We prove the claim below. Assuming it, \cref{lem:overview-transfer} applied to $a^\star$ gives
\[
 \rho\le C_1+C_2\sum_v\frac{\left(\sum_{e\ni v}p_ea_e^\star\right)^2}{q_v}
 \le C_1+\frac{2C_2b^2}{\rho}.
\]
Hence $\rho\le C_1+\sqrt{2C_2}\,b=O(1+\Theta+(b-1)T^2)$, using $T\ge1$.

Finally, the variance characterization of spectral independence uses the weights $p_e(1-p_e)$ in place of $p_e$. The insertion-odds bound gives $(1-p_e)^{-1}\le1+\Theta$, so this costs at most a factor $1+\Theta$. This proves \cref{thm:holant}, assuming \eqref{eq:proof-thm:holant-claim}.

It remains to prove \eqref{eq:proof-thm:holant-claim}. The vector $a^\star$ maximizes the variance subject to $\sum_ep_ea_e^2=1$. Differentiating the objective and constraint with respect to $a_e$ gives, for a Lagrange multiplier $\zeta$,
\[
 2\Cov\left(X_e,\sum_fa_f^\star X_f\right)=2\zeta p_ea_e^\star
\]
Multiplying by $a_e^\star/2$ and summing over $e$ identifies the multiplier:
\[
 \rho=\Var{}{\sum_ea_e^\star X_e}
 =\zeta\sum_ep_e(a_e^\star)^2=\zeta.
\]
Thus, for every edge $e$, the first-order optimality condition is
\begin{align} \label{eq:aux-optimal-cond}
 \Cov\left(X_e,\sum_fa_f^\star X_f\right)=\rho p_ea_e^\star.
\end{align}
For any vertex weights $h$, Cauchy--Schwarz and \eqref{eq:overview-degree} therefore imply
\begin{align*}
 \rho^2\left(\sum_vh_v\sum_{e\ni v}p_ea_e^\star\right)^2
 =\Cov\left(\sum_ea_e^\star X_e,\sum_vh_vD_v\right)^2
 \le\rho\Var{}{\sum_vh_vD_v}\le2\rho b^2\sum_vq_vh_v^2.
\end{align*}
If all vertex sums vanish, the claim is immediate. Otherwise, taking $h_v=q_v^{-1}\sum_{e\ni v}p_ea_e^\star$ and cancelling the positive factors gives \eqref{eq:proof-thm:holant-claim}.

\section{Mixing-time bounds for matching models}\label{sec:mixing}\label{app:matching-mixing}
We prove \cref{cor:matchings} by applying \cref{thm:field-comparison} and estimating the smallest positive configuration probability.

Write $\mu_s$ for the $b$-matching law at fugacity $s>0$. We use \cref{thm:small-fugacity-covariance} for the covariance bound at small fugacity and \cref{thm:field-comparison} to convert covariance bounds into a spectral gap.

\begin{proof}[Proof of \cref{cor:matchings}]
The spectral-independence assertions follow by substitution in \cref{thm:holant}. For monomer--dimer, take $b=1$ and $\gamma=\Theta=\lambda$, giving $O((1+\lambda)^2)$. For $b$-matchings at fugacity $\lambda>0$, take $\gamma=\Theta=\lambda$, so $T=\max\{1,\lambda^b\}$. Substitution gives the stated bound, which is $O(b)$ for uniform $b$-matchings ($\lambda=1$).

We next apply \cref{thm:field-comparison} to obtain the relaxation-time bounds. The $b$-matching law $\mu=\mu_\lambda$ has downward closed support, and every feasible insertion has weight ratio $\lambda$, so $r_{\mathrm{ins}}\le\lambda$. For every feasible occupied set $S$,
\[
 (1-\eta)*\mu^S=\mu_{\lambda(1-\eta)}^S.
\]
Thus it remains to bound the covariance of the residual matching law at fugacity $s=\lambda(1-\eta)$.

The proof of \cref{thm:holant} gives $\rho\le C_1+\sqrt{2C_2}\,b$. Using the explicit constants from \eqref{eq:edge-source}, we obtain
\[
 \Cov_{\mu_s^S}(X)\preceq A_{b,\lambda}\diag\bigl(\Ex_{\mu_s^S}X_e\bigr),
 \qquad 0<s\le\lambda,
\]
for every feasible occupied set $S$, where
\[
 A_{b,\lambda}=3+113\lambda+1800(b-1)\max\{1,\lambda^{2b}\}+4b.
\]
Indeed, residual capacities are at most $b$, and $T_b(s,s)=\max\{1,s^b\}\le\max\{1,\lambda^b\}$. Combining this estimate with \cref{thm:small-fugacity-covariance}, we verify the covariance hypothesis of \cref{thm:field-comparison} with
\[
 \eta_0=\max\left\{0,1-\frac{1}{2\lambda}\right\},\qquad
 c(\eta)=
 \begin{cases}
 A_{b,\lambda},&0<\eta<\eta_0,\\[2pt]
 \bigl[1-\lambda(1-\eta)\bigr]^{-1},&\eta_0\le\eta<1.
 \end{cases}
\]
The second branch applies because $\lambda(1-\eta)\le1/2$ there. The integral in \cref{thm:field-comparison} satisfies
\begin{align*}
 \int_0^1\frac{c(\eta)-1}{1-\eta}\,d\eta
 &=(A_{b,\lambda}-1)\log\frac{1}{1-\eta_0}
   -\log\bigl(1-\lambda(1-\eta_0)\bigr)\\
 &\le(A_{b,\lambda}-1)\log\max\{1,2\lambda\}+\log2.
\end{align*}
Hence \cref{thm:field-comparison} gives
\[
 \gap(P_{\mathrm{GD}})
 \ge\frac{1}{2(1+\lambda)\max\{1,2\lambda\}^{A_{b,\lambda}-1}m}
 =\Omega_{b,\lambda}(m^{-1}).
\]
Since heat-bath dynamics is positive semidefinite, $t_{\mathrm{rel}}=1/\gap(P_{\mathrm{GD}})=O_{b,\lambda}(m)$. For monomer--dimer, setting $b=1$ gives $A_{1,\lambda}=7+113\lambda$, and therefore $t_{\mathrm{rel}}=O_\lambda(m)$.

For uniform $b$-matchings, $\lambda=1$ and $A_{b,1}\le1804b$. Writing $A=1804b$, the same covariance estimates verify the hypothesis of \cref{thm:field-comparison} with $c(\eta)=\min\{A,\eta^{-1}\}$. Now
\begin{align*}
 \int_0^1\frac{c(\eta)-1}{1-\eta}\,d\eta
 &=-(A-1)\log(1-A^{-1})+\log A
 &\le1+\log A.
\end{align*}
Since $r_{\mathrm{ins}}\le1$, \cref{thm:field-comparison} gives
\[
 \gap(P_{\mathrm{GD}})\ge\frac{1}{2\mathrm e A m},
 \qquad t_{\mathrm{rel}}\le2\mathrm e A m=O(bm).
\]

Finally, let $\mu_{\min}$ be the smallest positive configuration probability. By \cref{prop:spectral-gap_mixing-time}, random-scan heat bath satisfies
\[
 t_{\mathrm{mix}}(\varepsilon)
 =O\left(t_{\mathrm{rel}}\left(\log\frac1{\mu_{\min}}
                                  +\log\frac1\varepsilon\right)\right).
\]
We now apply this bound to the two models. A matching can be encoded by listing each vertex's partner or recording that it is unmatched, so $|\mathcal M(G)|\le(n+1)^n$. Also,
\[
 \mu_{\min}^{-1}
 \le|\mathcal M(G)|\max\{\lambda,\lambda^{-1}\}^{n/2},
\]
which gives $\log(1/\mu_{\min})=O_\lambda(n\log n)$.

For $b$-matchings at fugacity $\lambda>0$, every feasible configuration has at most $k=\min\{m,bn/2\}$ edges, so
\[
 \mu_{\min}^{-1}\le |\mathcal M_{\mathbf b}|\max\{\lambda,\lambda^{-1}\}^{k}.
\]
There are at most $2^m$ edge subsets. Alternatively, encoding each vertex's selected neighbors in at most $b$ slots gives $|\mathcal M_{\mathbf b}|\le(n+1)^{bn}$. Hence
\[
 \log\frac1{\mu_{\min}}
 \le\min\{m\log2,bn\log(n+1)\}
       +\min\{m,bn/2\}|\log\lambda|.
\]
For fixed $\lambda>0$, this is $O_\lambda(\min\{m,bn\log n\})$. Combining these estimates with the relaxation-time bounds proves \eqref{eq:monomer-relaxation} and \eqref{eq:b-matching-relaxation}.

For uniform $b$-matchings, setting $\lambda=1$ removes the $|\log\lambda|$ term and gives
\[
 \log\frac1{\mu_{\min}}=O\bigl(\min\{m,bn\log n\}\bigr)
\]
with an absolute implied constant. Together with $t_{\mathrm{rel}}=O(bm)$, this proves \eqref{eq:uniform-relaxation}.
\end{proof}

\end{document}